\newenvironment{symbolfootnotes}
  {\par\edef\savedfootnotenumber{\number\value{footnote}}
   
   \setcounter{footnote}{1}}
  {\par\setcounter{footnote}{\savedfootnotenumber}}

\newcommand{\doublefigure}[8]{
\begin{minipage}[t]{\textwidth}
\fbox{
\begin{minipage}[t][#1][b]{.58\textwidth}
\vspace{0pt}
\centering

 \includegraphics[width=#5\textwidth]{#3.eps}
\medskip

\begin{minipage}[t][#2][t]{\textwidth}
\centering\small
(a) #4
\end{minipage}
\end{minipage}}
\fbox{\begin{minipage}[t][#1][b]{0.38\textwidth}
\centering
 \includegraphics[width=#8\textwidth]{#6.eps}
\medskip

\begin{minipage}[t][#2][t]{\textwidth}
\centering\small
(b) #7
\end{minipage}
\end{minipage}}
\end{minipage}
}

\newcommand{\doublefigurestacked}[9]{
\fbox{
\begin{minipage}[t][#1][b]{0.98\textwidth}
\centering
 \includegraphics[width=#8\textwidth]{#6.eps}
\medskip

\begin{minipage}[t][#9][t]{\textwidth}
\centering\small
(a) #7
\end{minipage}
\end{minipage}}

\begin{minipage}[t]{0.98\textwidth}
\fbox{
\begin{minipage}[t][45mm][b]{\textwidth}
\vspace{0pt}
\centering
 \includegraphics[width=0.6\textwidth]{#3.eps}
\medskip

\begin{minipage}[t][#2][t]{\textwidth}
\centering\small
(b) #4
\end{minipage}
\end{minipage}}

\end{minipage}
}

\documentclass[10pt]{article}

\usepackage{graphicx}
\usepackage{cite}
\usepackage{hyperref}
\usepackage{amsmath}
\usepackage{amssymb}
\usepackage{amsfonts}
\usepackage{amsthm}
\usepackage{fullpage}
\usepackage{multirow}
\usepackage{amsthm}
\usepackage{subfig}
\usepackage{sidecap}

\input{commands-tam.sty}

\theoremstyle{definition}
\newtheorem{lemma}{Lemma}[section]
\newtheorem{theorem}{Theorem}[section]

\begin{document}

\title{\textbf{Optimal program-size complexity for self-assembly at temperature 1 in 3D}}

\author{%
David Furcy\thanks{Computer Science Department, University of Wisconsin--Oshkosh, Oshkosh, WI 54901, USA,\protect\url{furcyd@uwosh.edu}.}
\and
Samuel Micka\thanks{Computer Science Department, University of Wisconsin--Oshkosh, Oshkosh, WI 54901, USA,\protect\url{mickas37@gmail.com}.}
\and
Scott M. Summers\thanks{Computer Science Department, University of Wisconsin--Oshkosh, Oshkosh, WI 54901, USA,\protect\url{summerss@uwosh.edu}.}
}


\date{}
\maketitle

\begin{abstract}
Working in a three-dimensional variant of Winfree's abstract Tile Assembly Model, we show that, for all $N \in \mathbb{N}$, there is a tile set that uniquely self-assembles into an $N \times N$ square shape at temperature 1 with optimal program-size complexity of $O(\log N / \log \log N)$ (the program-size complexity, also known as tile complexity, of a shape is the minimum number of unique tile types required to uniquely self-assemble it). Moreover, our construction is ``just barely'' 3D in the sense that it works even when the placement of tiles is restricted to the $z = 0$ and $z = 1$ planes. This result affirmatively answers an open question from Cook, Fu, Schweller (SODA 2011). To achieve this result, we develop a general 3D temperature 1 optimal encoding construction, reminiscent of the 2D temperature 2 optimal encoding construction of Soloveichik and Winfree (SICOMP 2007), and perhaps of independent interest.
\end{abstract} 

\newpage

\section{Introduction}

The simplest mathematical model of nanoscale tile self-assembly is
Erik Winfree's abstract Tile Assembly Model (aTAM) \cite{Winf98}. The
aTAM extends classical Wang tiling \cite{Wang61} in that the former
bestows upon the latter a mechanism for sequential ``growth'' of a
tile assembly. Very briefly, in the aTAM, the fundamental components
are un-rotatable, translatable square ``tile types'' whose sides are
labeled with (alpha-numeric) glue ``colors'' and (integer)
``strengths''. Two tiles that are placed next to each other
\emph{bind} if both the glue colors and the strengths on their
abutting sides match and the sum of their matching strengths sum to at
least a certain (integer) ``temperature''. Self-assembly starts from a
``seed'' tile type, typically assumed to be placed at the origin, and
proceeds nondeterministically and asynchronously as tiles bind to the
seed-containing assembly one at a time. In this paper, we work in a
three-dimensional variant of the aTAM in which tile types are unit
cubes and growth proceeds in a \emph{noncooperative} manner.

Tile self-assembly in which tiles may be placed in a noncooperative
fashion is often referred to as ``temperature~1
self-assembly''. Despite the arcane name, this is a fundamental and
ubiquitous form of growth: it refers to growth from \emph{growing and
  branching tips} in Euclidean space, where each new tile is added if
it can bind on at least \emph{one side}. Note that a more general
form of \emph{cooperative} growth, where some of the tiles may be
required to bind on two or more sides, leads to highly non-trivial
behavior in the aTAM, e.g., Turing universality \cite{Winf98} and  the
efficient self-assembly of $N \times N$ squares
\cite{AdlemanCGH01,RotWin00} and other algorithmically specified
shapes \cite{SolWin07}. Doty, Patitz and Summers conjecture
\cite{jLSAT1} that the shape or pattern produced by any 2D temperature
1 tile set that uniquely produces a final structure is ``simple'' in
the sense of Presburger arithmetic \cite{Presburger30}. However, their
conjecture is currently unproven and it remains to be seen if
noncooperative self-assembly in the aTAM can achieve the same
computational and geometric expressiveness as that of cooperative
self-assembly. In this paper, we specifically focus on a problem that
is very closely related to that of finding the minimum number of
distinct tile types required to self-assemble an $N \times N$ square,
i.e., its \emph{tile complexity} (or \emph{program-size complexity}), at temperature 1.

The tile complexity of an $N \times N$ square at temperature 1 has
been studied extensively. In 2000, Rothemund and Winfree
\cite{RotWin00} proved that the tile complexity of an $N \times N$
square at temperature 1 is $N^2$, assuming the final structure is
fully connected, and at most $2N - 1$, otherwise (they also
conjectured that the lower bound, in general, is $2N - 1$). A decade
later, Manuch, Stacho and Stoll \cite{ManuchSS10} established that,
assuming no mismatches are present in the final assembly, the tile
complexity of an $N \times N$ square at temperature 1 is $2N -
1$. Shortly thereafter, and quite surprisingly, Cook, Fu and Schweller
\cite{CooFuSch11} showed that the tile complexity of an $N \times N$
square at temperature 1 is $O(\log N)$ if tiles are
allowed to be placed in the $z = 0$ and $z = 1$ planes (here, an $N
\times N$ square is actually a full 2D square in the $z=0$ plane with
additional tiles above it in the $z = 1$ plane).

Technically speaking, the aforementioned, just-barely-3D construction
of Cook, Fu and Schweller is actually a general transformation that
takes as input a 2D temperature 2 ``zig-zag'' tile set, say $T$, and
outputs a corresponding 3D temperature 1 tile set, say $T'$, that
simulates $T$. In this transformation from $T$ to $T'$, the tile
complexity 
increases by $O(\log g)$,
where $g$ is the number of unique north/south glues in the input tile
set $T$. Since the number of north/south glues in the standard 2D aTAM
base-$2$ binary counter is $O(1)$, Cook, Fu and Schweller use their
transformation to produce several tile sets, which, when wired
together appropriately and combined with ``filler'' tiles,
self-assemble into an $N \times N$ square at temperature 1 in 3D with
$O(\log N)$ tile complexity.

Of course, it is well-known that the tile complexity of an $N \times N$ square at temperature 2 is $O\left(\frac{\log N}{\log \log N}\right)$ \cite{AdlemanCGH01}, which, as Cook, Fu and Schweller point out in \cite{CooFuSch11}, is achievable using a zig-zag counter with an optimally-chosen base, say $b$, which satisfies $\frac{\log N}{\log \log N} < b < \frac{2 \log N}{\log \log N}$, rather than in base $b = 2$. However, using currently-known techniques, counting in base $b$ at temperature 2 requires having a tile set with $\Theta(b)$ unique north/south glues, whence the zig-zag transformation of Cook, Fu and Schweller cannot be used to get $O\left(\frac{\log N}{\log \log N}\right)$ tile complexity for an $N \times N$ square at temperature 1 in 3D. Moreover, the \emph{optimal encoding} scheme of Soloveichik and Winfree \cite{SolWin07} and the \emph{base conversion} technique of Adleman et. al. \cite{AdlemanCGH01} do not work correctly at temperature 1 and they also cannot be simulated by the Cook, Fu and Schweller construction without an $\Omega\left(\frac{\log N}{\log\log N}\right)$ blowup in tile complexity. Thus, Cook, Fu and Schweller, at the end of section 4.4 in \cite{CooFuSch11}, pose the following question: is it possible to achieve the tile complexity bound of $O\left(\frac{\log N}{\log \log N}\right)$ for an $N \times N$ square at temperature $1$ in 3D?

In Theorem~\ref{thm:main-theorem}, the main theorem of this paper, we answer the previous question in the affirmative, i.e., we prove that the tile complexity of an $N \times N$ square at temperature 1 in 3D is $O\left(\frac{\log N}{\log \log N}\right)$ (in our construction, tiles are placed only in the $z = 0$ and $z = 1$ planes of $\Z^3$). Our tile complexity matches a corresponding lower bound dictated by Kolmogorov complexity (see \cite{Li:1997:IKC} for details on Kolmogorov complexity), which was established by Rothemund and Winfree in 2000, and holds for all ``algorithmically random'' values of $N$ \cite{RotWin00}\footnote{Technically, Rothemund and Winfree established the 2D self-assembly case, but their proof easily generalizes to 3D self-assembly.}. Thus, our construction yields optimal tile complexity for the self-assembly of $N \times N$ squares at temperature 1 in 3D, for all algorithmically random values of $N$. To achieve optimal tile complexity, we adapt the optimal encoding technique of Soloveichik and Winfree \cite{SolWin07} (which, itself, is based on the base-conversion scheme of \cite{AdlemanCGH01}) to work at temperature 1 in 3D. Our 3D temperature 1 optimal encoding technique, described in Section~\ref{sec:optimal_encoding}, is perhaps of independent interest.

\section{Definitions}\label{sec-definitions}
In this section, we give a brief sketch of a $3$-dimensional version of Winfree's abstract Tile Assembly Model.

\subsection{3D abstract Tile Assembly Model}

Let $\Sigma$ be an alphabet. A $3$-dimensional \emph{tile type} is a tuple $t \in (\Sigma^* \times \N)^{6}$, e.g., a unit cube with six sides listed in some standardized order, each side having a \emph{glue} $g \in \Sigma^* \times \N$ consisting of a finite string \emph{label} and a non-negative integer \emph{strength}. In this paper, all glues have strength $1$.
There is a finite set $T$ of $3$-dimensional tile types but an infinite number of copies of each tile type, with each copy being referred to as a \emph{tile}.

A $3$-dimensional \emph{assembly} is a positioning of tiles on the
integer lattice $\Z^3$ and is described formally as a partial function
$\alpha:\Z^3 \dashrightarrow T$. Two adjacent tiles in an assembly
\emph{bind} if the glue labels on their abutting sides are equal and
have positive strength.  Each assembly induces a \emph{binding graph},
i.e., a grid graph whose vertices are (positions of) tiles and whose
edges connect any two vertices whose corresponding tiles bind.  If
$\tau$ is an integer, we say that an assembly is
\emph{$\tau$-stable} if every cut of its binding graph has strength at
least~$\tau$, where the strength of a cut is the sum of all of the
individual glue strengths in~the~cut.

A $3$-dimensional \emph{tile assembly system} (TAS) is a triple $\calT
= \left(T,\sigma,\tau\right)$, where $T$ is a finite set of tile
types, $\sigma:\Z^3 \dashrightarrow T$ is a finite, $\tau$-stable
\emph{seed assembly}, and $\tau$ is the \emph{temperature}. In this
paper, we assume that $|\dom{\sigma}| = 1$ and $\tau=1$. An assembly
$\alpha$ is \emph{producible} if either $\alpha = \sigma$ or if
$\beta$ is a producible assembly and $\alpha$ can be obtained from
$\beta$ by the stable binding of a single tile.  In this case we write
$\beta\to_1^{\calT} \alpha$ (to mean~$\alpha$ is producible from
$\beta$ by the binding of one tile), and we write $\beta\to^{\calT}
\alpha$ if $\beta \to_1^{\calT^*} \alpha$ (to mean $\alpha$ is
producible from $\beta$ by the binding of zero or more tiles).
When~$\calT$ is clear from context, we may write $\to_1$ and $\to$
instead.  We let $\mathcal{A}\left[\mathcal{T}\right]$ denote the set
of producible assemblies of~$\calT$.  An assembly is \emph{terminal}
if no tile can be $\tau$-stably bound to it.  We
let~$\mathcal{A}_{\Box}\left[\mathcal{T}\right] \subseteq
\mathcal{A}\left[\mathcal{T}\right]$ denote the set of
producible, terminal assemblies of $\calT$.

A TAS~$\calT$ is \emph{directed} if $\left|\mathcal{A}_{\Box}\left[\calT\right]\right| = 1$. Hence, although a directed system may be nondeterministic in terms of the order of tile placements,  it is deterministic in the sense that exactly one terminal assembly is producible. For a set $X \subseteq \Z^3$, we say that $X$ is uniquely produced if there is a directed TAS $\mathcal{T}$, with $\mathcal{A}_{\Box}\left[\calT\right] = \{\alpha\}$, and $\dom{\alpha} = X$.

For $N \in \mathbb{N}$, we say that $S^3_N \subseteq \mathbb{Z}^3$ is a 3D $N \times N$ \emph{square} if $\{0,\ldots,N-1\} \times \{0,\ldots,N-1\} \times \{0\} \subseteq S^3_N \subseteq \{0,\ldots,N-1\} \times \{0,\ldots,N-1\} \times \{0,1\}$. In other words, a 3D $N \times N$ square is at most two 2D $N \times N$ squares, one stacked on top of the other. \begin{symbolfootnotes}
In the spirit of \cite{RotWin00}, we define the \emph{tile complexity} of a 3D $N \times N$ square at temperature $\tau$, denoted by $K^\tau_{3DSA}(N)$, as the minimum number of distinct 3D tile types required to uniquely produce it, i.e., $K^\tau_{3DSA}(N) = \min \left\{ n   \; \left| \; \mathcal{T} = \left(T,\sigma,\tau\right), \left|T\right|=n \textmd{ and } \mathcal{T} \textmd{ uniquely produces } S^3_N\right.\right\}\footnote{One subtle difference between our 3D definition of $K$ and the original 2D definition of the tile complexity of an $N \times N$ square, given by Rothemund and Winfree in \cite{RotWin00}, is that they assume a fully-connected final structure, whereas we do not.}$.
\end{symbolfootnotes}

\subsection{Notation for figures}

In the figures in this paper, we use big squares to represent tiles placed in the $z=0$ plane and small squares to represent tiles placed in the $z=1$ plane. A glue between a $z=0$ tile and $z=1$ tile is denoted as a small black disk. Glues between $z=0$ tiles are denoted as thick lines. Glues between $z=1$ tiles are denoted as thin lines.

\section{Optimal encoding at temperature 1}
\label{sec:optimal_encoding}

A key problem in algorithmic self-assembly is that of \emph{providing input to a tile assembly system} (e.g., the size of a square, the input to a Turing machine, etc.). In real-world laboratory implementations, as well as theoretical constructions, input to a tile system is typically provided via a (possibly large) collection of ``hard-coded'' seed tile types that uniquely assemble into a convenient ``seed structure,'' such as a line of tiles that encodes some input value. Unfortunately, in practice, it is more expensive to manufacture different types of tiles than it is to create copies of each tile type. Thus, it is critical to be able to provide input to a tile system using the smallest possible number of hard-coded seed tile types.

Consider the problem of constructing a tile set that uniquely self-assembles from a single seed tile into a ``seed row'' that encodes an $n$-bit binary string, say $x$. The most straightforward way to do this is to construct a set of $n$ unique tile types that deterministically assemble into a line of tiles of length $n$, where each tile in the line represents a different bit of $x$. This simple construction encodes one bit of $x$ per tile, whence its tile complexity is $O(n)$. Note that, in this example, each tile type is an element of a set of size $n$, yet each tile type encodes only 1 bit of information, instead of the optimal $O(\log n)$ bits. Is there a more efficient encoding construction?


The optimal encoding constructions of Adleman et al. \cite{AdlemanCGH01}, and Soloveichik and Winfree \cite{SolWin07} are more efficient methods of encoding input to a tile set. These constructions are based on the idea that each seed row tile type should encode $k = O(\log n)$ bits -- instead of a single bit -- of $x$, which means that $O(n / \log n)$ unique tile types suffice to uniquely self-assemble into a seed row that encodes the bits of $x$. Unfortunately, now the bits of $x$ are no longer conveniently represented in distinct tiles. Fortunately, if $k$ is chosen carefully, then it is possible to use a tile set of size $O(n / \log n)$ to ``extract'' the bits of $x$ into a more convenient one-bit-per-tile representation, which can be used to seed a binary counter or a Turing machine simulation.

Up until now, all known optimal encoding constructions (e.g., \cite{AdlemanCGH01, SolWin07}) required cooperative binding (that is, temperature $\tau \geq 2$). In what follows, we propose an optimal encoding construction (based on the construction of Soloveichik and Winfree \cite{SolWin07}) that works at temperature $\tau=1$ and is ``just barely'' 3D, i.e., tiles are only placed in the $z=0$ and $z=1$ planes.

\subsection{Setup}

Let $x = x_{n-1}x_{n-2}...x_1x_0$ be the input string, where $x_i \in \{0,1\}$. Let $m = \left\lceil n / k \right\rceil$, where $k$ is the smallest integer satisfying $2^k \geq n/\log n$. We write $x = w_{0}w_{1}...w_{m-2}w_{m-1}$, where each $w_i$ is a $k$-bit block. Note that $w_{0}$ is padded to the left with leading 0's, if necessary. In the figures in this section, a green tile represents a starting point for some portion of an assembly sequence, and a red tile represents an ending point.


\subsection{Overview of the construction}

We extract each of the $m$ $k$-bit blocks within a roughly rectangular
region of space of width $O(k)$ and height $O(m)$.  We
refer to this region of space as a ``block extraction region'' (or simply ``extraction region'').
For each $0 \leq i < m$,
we extract block $w_i$ in extraction region $i$.
Each extraction region, other than the first and last ones, assembles via a series of gadgets (small groups of tiles
that carry out a specific task).

We encode the $k$ bits of a $k$-bit block as a series of
geometric bumps along a path of tiles that makes up the top
border of an extraction region. A bump in the $z=0$ plane represents
the bit 0 and a bump in the $z=1$ plane represents the bit 1.
The end result of our construction is an assembly in which each bit of $x$ is encoded in its own bit-bump (see Figure~\ref{fig:twelve} for an example).

We extract the $k$-bit blocks
in order, starting with the first block $w_{0}$, which represents
the most significant bits of $x$. Normally, to carry out this
sort of activity at temperature 1 (i.e., to enforce the ordering of
tile placements), one has to encode the order of placement directly into
the glues of the tiles. However, for our construction, this would essentially mean
encoding the number of the block that is being extracted
into the glues of the tiles that fill in its extraction region. Unfortunately, doing so, at least in
the most straightforward way, results
in an increase in tile complexity from the optimal $O(n / \log n)$ to
$\Omega(n^2 / \log^2 n)$.

Therefore, in our construction, we encode the number of the block that
is being extracted as a geometric pattern along a vertical path of
tiles that runs along the right side of each extraction region. We
call this special geometric pattern the ``block number.'' Then we use a
special gadget called the ``block-number gadget'' to search for this
pattern.

\begin{SCfigure}
\centering
 \includegraphics[width=0.4\textwidth]{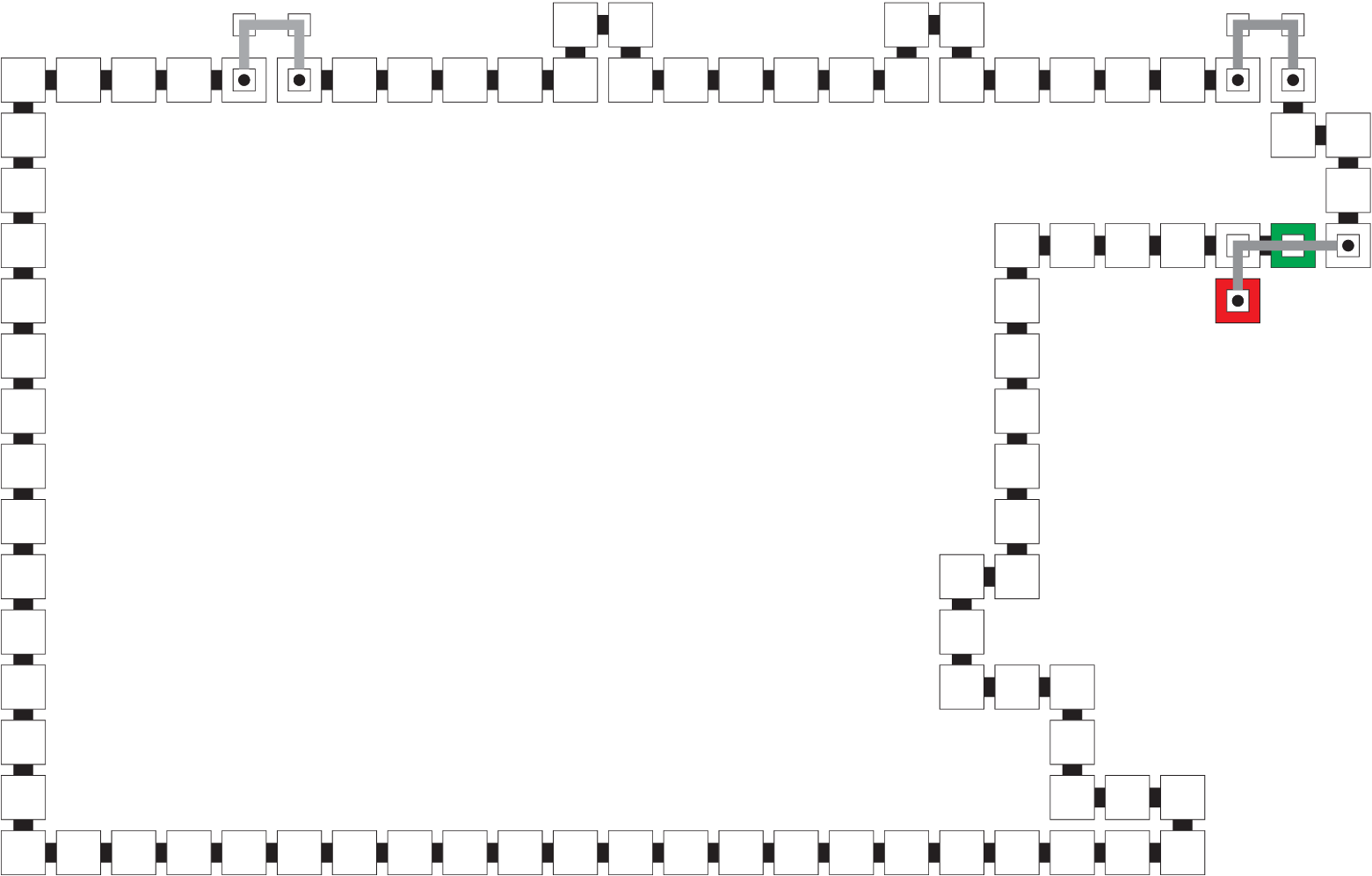}
\caption{ The perimeter of the first extraction region is hard-coded to self-assemble like this. In this example, the four bumps along the top (from left to right) represent the bits 1, 0, 0 and 1, respectively. The green tile (bottom tile in the penultimate column) is the single seed tile for our entire optimal encoding construction.  }
\label{fig:first}
\vspace{0pt}
\end{SCfigure}

Within an extraction region, the block number determines which block
gets extracted next. Basically, the path along which the block number
is encoded blocks the placement of $m - 1$ special tiles, each
of which tries to initiate the extraction of a particular $k$-bit
block. We call these special tiles ``extraction tiles.'' Since the
first extraction region is hard-coded (see below), the first block
does not have an extraction tile associated with it.  Within any given extraction
region, exactly one extraction tile will not be blocked. The one extraction tile that is not blocked by the block number gadget will initiate
the extraction of the $k$ bits of the block to which it corresponds.

\begin{figure}[htp]%
\centering
    \subfloat[][The block-number gadget determines the next block to extract by ``searching'' for the position of the block number (i.e., the position of the notch in which the red tile is ultimately placed).]{%
        \label{fig:square0}%
        \includegraphics[width=1.5in]{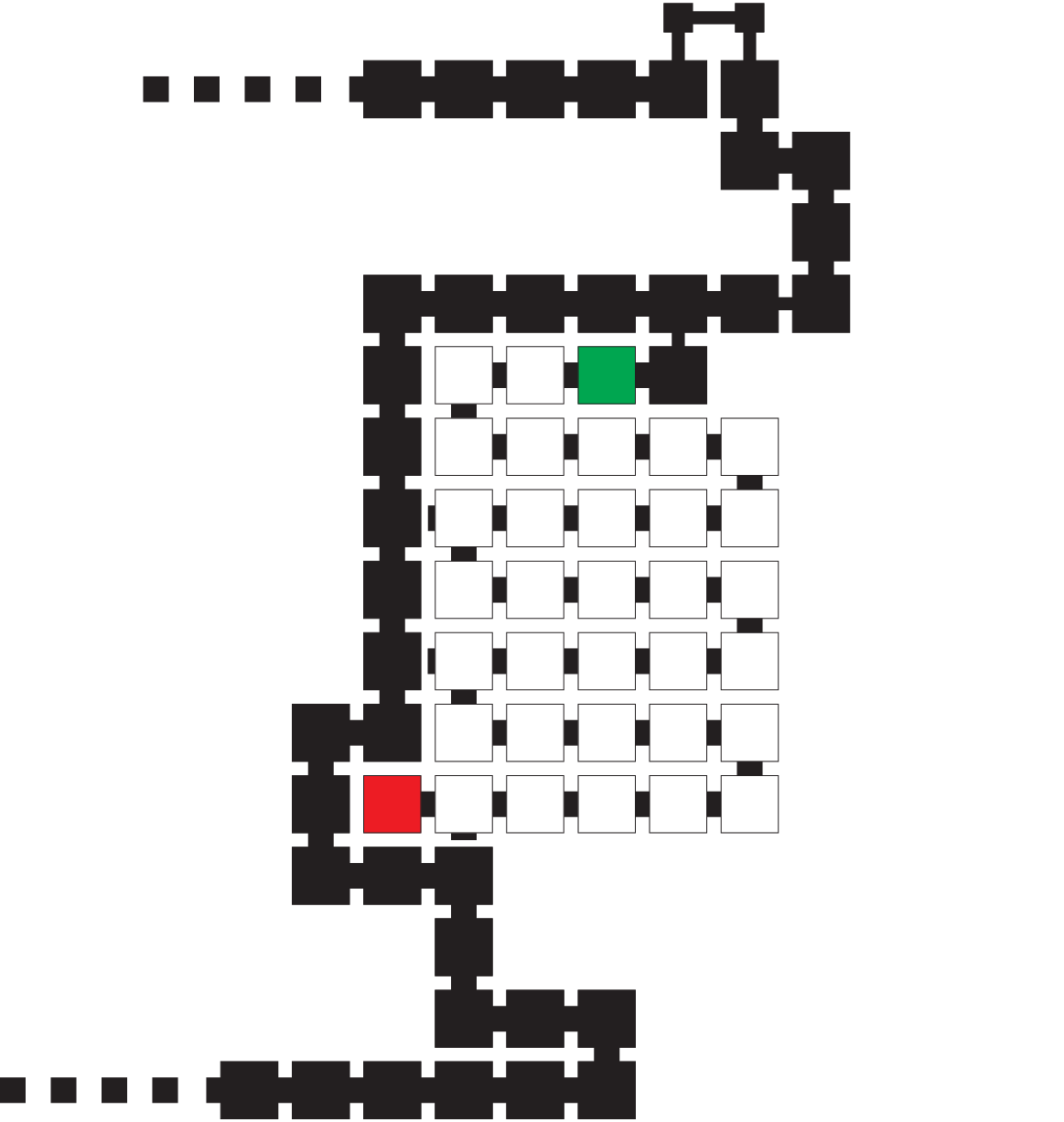}}%
        \hspace{10pt}%
    \subfloat[][The path initiated by the extraction tile for $w_1$ ``jumps'' over the block-number gadget and grows a hook to block a subsequent gadget.]{%
        \label{fig:square1}%
        \includegraphics[width=1.5in]{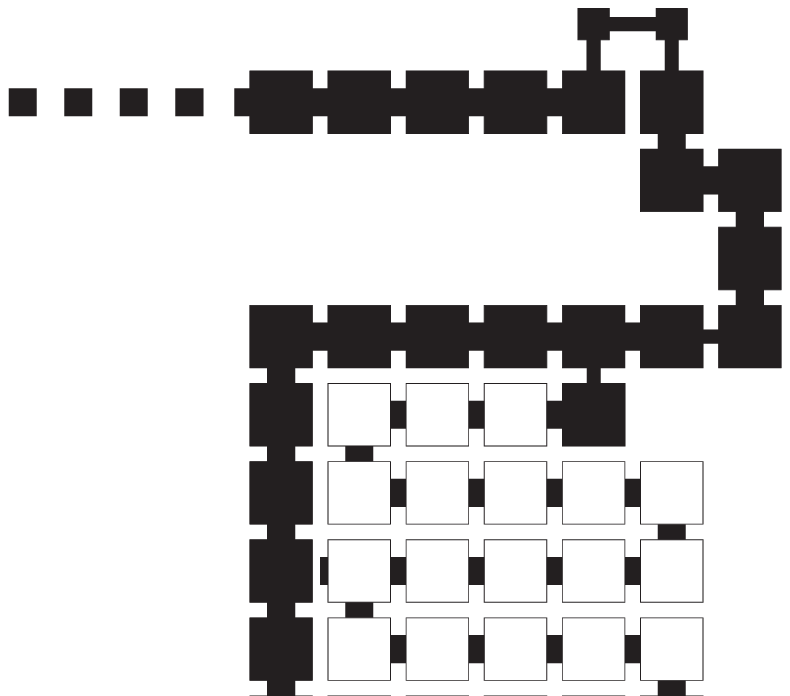}}
        \hspace{10pt}
    \subfloat[][The path initiated by the extraction tile for $w_1$ continues growing upward and eventually finds the top of the block-number gadget. The upward growth of this path is blocked by a portion of the previous extraction region.]{%
        \label{fig:square2}%
        \includegraphics[width=1.5in]{}}
        \hspace{10pt}
    \subfloat[][Once at the top of the block-number gadget, the path initiated by the extraction tile for $w_1$ ``jumps'' over a portion of the previous extraction region and starts extracting the bits of $w_1$ along the top of the second extraction region.]{%
        \label{fig:square3}%
        \includegraphics[width=1.5in]{}}
    \caption{\label{fig:block_path} This sequence of figures shows how the position of the block number is found. The black tiles correspond to tiles of the previous extraction region. }%
\end{figure}

In our construction, we hard-code the assembly of the first and last
extraction regions.  What this means is that, in each of these
extraction regions, a single-tile-wide path assembles the perimeter
and then we use a filler tile to fill in the interior. For this step, it is crucial to first assemble the perimeter of the extraction region and then use a filler tile to tile the interior. Note that, if one were to uniquely tile every location in the first (or last) extraction region, then the tile complexity of the construction would be $\Omega(mk)$, which is not optimal. Tiling the perimeter of either the first or last extraction region can be done with $O(m + k)$ unique tile types (see
Figure~\ref{fig:first} for the example of the first extraction region).

All extraction regions other than the first and last ones are
constructed using a general set of gadgets. In the second extraction
region, which is the first generally-constructed extraction region,
the block-number gadget determines that $w_1$ is the next block to be
extracted by ``searching'' for the block number position. When the block number is
found, a path of tiles, initiated by the extraction tile for $w_1$, is allowed to assemble (see
Figure~\ref{fig:block_path} for an example of this process). In
general, for extraction region $i$, for all $1 \leq i < m - 1$, the
path along which the block number is encoded geometrically hinders
the placement of all extraction tiles that correspond to blocks
$w_1, ..., w_{i-1}, w_{i+1}, ..., w_{m-2}$.

Each extraction tile initiates the extraction of the $k$-bit block to which it corresponds (see Figure~\ref{fig:five}). We use a set of ``bit-extraction'' gadgets to extract a $k$-bit block into a one-bit-per-bump representation (the bit extraction gadgets are collectively referred to as the ``extraction gadget''). Our bit extraction gadgets are basically 3D, temperature 1 versions of the ``extract bit'' tile types in Figure 5.7a of \cite{SolWin07}.

\begin{SCfigure}
\includegraphics[width=.6\textwidth]{}
\caption{\label{fig:five} The bits of the current block are represented as bumps along the top of the extraction region that is currently being assembled. }
\end{SCfigure}

After a block, say $w_{i}$, for $i > 0$, is extracted, the block
number is geometrically ``incremented'', i.e., its position is
translated up by a small constant amount (notice the position of the
white ``hook'' at the bottom of Figure~\ref{fig:five}). We do this in
two phases. First, the current position of the block number is found
and then it is incremented and translated. Figure~\ref{fig:six} shows
how the current position of the block number is detected using a
zig-zag path of tiles. Figure~\ref{fig:seven} shows how the current
position of the block number is geometrically incremented.

\begin{SCfigure}
\includegraphics[width=.6\textwidth]{}
\caption{\label{fig:six} A path of tiles searches for the block number, represented by a notch in a previous portion of the assembly. The red tile ``knows'' that it found the position of the block number because it was allowed to be placed.}
\end{SCfigure}

\begin{SCfigure}
\includegraphics[width=.6\textwidth]{}
\caption{\label{fig:seven} The block number is geometrically incremented. The green tile ``jumps'' over the previous gadget that found the position of the block number and grows a hook of tiles to represent the updated block number. Notice that the new hook of tiles is two tiles higher than the previous hook (shown in black), which corresponds to the two rows of tiles that each block takes up in the block-number gadget.}
\end{SCfigure}

After the block number has been updated, a series of gadgets geometrically propagate the position of the block number to the right through the remainder of extraction region $i$ so that it is advertised to extraction region $i + 1$. This is shown in Figures~\ref{fig:eight} and~\ref{fig:nine}. Technically, we geometrically propagate the block number position through the rest of the extraction region using a series of gadgets. Logically, however, we do this in two phases, which are iterated: ``up'' propagation and ``down'' propagation.

\begin{SCfigure}
\includegraphics[width=.6\textwidth]{}
\caption{\label{fig:eight} A series of gadgets geometrically propagate
  the position of the block number through the rest of the extraction
  region. This figure shows two of the gadgets. The first one
  assembles upward until it is blocked by a previous portion of the
  assembly. The second one assembles horizontally and to the right
as it jumps over the top row of the previous gadget.}
\end{SCfigure}

\begin{SCfigure}
\includegraphics[width=.6\textwidth]{}
\caption{\label{fig:nine} The position of the block number is propagated through the rest of the extraction region.}
\end{SCfigure}

The ``up'' propagation phase grows from the position of the block number up to (and is blocked by) a previous portion of the assembly. This is shown in Figure~\ref{fig:eight}. The ``down'' propagation phase grows from the top of the previous (up) propagation phase back down to the position of the block number. The upward growth of each up propagation phase is blocked in the $z=0$ plane but not in the $z=1$ plane. However, this is switched for the last up propagation phase. In other words, the last up propagation phase may continue its upward growth, which signals the end of the extraction region, but its $z=1$ growth is blocked. In Figure~\ref{fig:ten}, the last up propagation phase is allowed to continue its upward growth in the $z=0$ plane.

\begin{SCfigure}
\includegraphics[width=.6\textwidth]{}
\caption{\label{fig:ten} The last up propagation phase detects when it has reached the end of the extraction region and initiates a perimeter gadget (see Figure~\ref{fig:eleven}) that will fill in the bottom row of the current extraction region before the next extraction region begins. }
\end{SCfigure}

The last up propagation phase initiates the assembly of a special gadget that fills in the bottom row of the current extraction region before the next extraction region begins. The reason we do this is to ensure that, when the entire extraction process is done (i.e., when all $n$ bits have been extracted into a one-bit-per-bump representation), the bottom row of the assembly is completely filled in. Figure~\ref{fig:eleven} shows an example of how this gadget tiles the remaining perimeter of an extraction region. Note that the tile complexity of this gadget is the size of the perimeter of an extraction region, i.e., $O(m)$.

\begin{SCfigure}
\includegraphics[width=.6\textwidth]{}
\caption{\label{fig:eleven} The bottom row of the extraction region is tiled by a special gadget with $O(m)$ tile complexity. After the bottom row of the extraction region is tiled, the next extraction region is initiated. Notice that the red tile in this figure belongs to the same row of tiles as the red tile in Figure~\ref{fig:first} but the position of the block number has moved up, which means the extraction tile for the next block (in this case, $w_2$) will be allowed to assemble and all other extraction tiles will be blocked. }
\end{SCfigure}

The final extraction region, like the initial extraction region, is
hard-coded to assemble its perimeter via a single-tile-wide path.  The
tiles that comprise the final extraction region ``know'' to stop the
extraction process and possibly initiate the growth of some other
logical component of a larger assembly, e.g., a binary counter or a Turing
machine simulation in which the extracted bits of $x$, along the top
of each of the $m$ extraction regions, are used as input.

The end result of our optimal encoding construction is a roughly rectangular assembly of tiles with height $O(m)$ and width $O(n)$, where each bit of $x$ is encoded as a bump (either in the $z=0$ or $z=1$ plane) along the top of the rectangle, with four ``spacer'' tiles to the left and right of each bit-bump. Figure~\ref{fig:twelve} shows the result of our optimal encoding construction with four extraction regions. Complete details for this construction are in the appendix (see Section~\ref{app:optimal_encoding}).

\begin{figure}[htp]
\includegraphics[width=\textwidth]{}
\caption{\label{fig:twelve} This is an example of our optimal encoding construction using $n=16$ and $k=4$. Note that this does not correspond to an actual instance of our optimal encoding construction because if $n=16$, then the smallest value of $k$ satisfying $2^k \geq n/\log n$ is $k=2$. The bit string encoded along the top is $1001001100111000$. All of the empty spaces in the $z=0$ plane are filled in with the same filler tile. }
\end{figure}

\subsection{Tile complexity}\label{sec:OER_tile_complexity}
To establish the tile complexity bound of $O(n/ \log n)$ for our construction, we use the following technical lemma.

\begin{lemma}\label{lem:m_technical_lemma} Let $1 < n \in \Z^+$ and $m = \left\lceil n / k\right\rceil$, where $k$ is the smallest integer satisfying $2^k \geq n/\log n$. Then $m = O(n / \log n)$.
\end{lemma}

\begin{proof}
Assume the hypothesis. First, note that, by our choice of $k$, $2^k < 2n / \log n$. Then we have
\begin{eqnarray*}
m & = & \left\lceil \frac{n}{k} \right\rceil \leq \left\lceil n/\log(n/\log n) \right\rceil  = \left\lceil n/(\log n - \log \log n) \right\rceil < \left\lceil n/(\log n - (\log n) / 2) \right\rceil \\	
  & \leq & 2n/\log n + 1 \leq 2n/\log n + n/\log n = 3n/\log n = O(n / \log n)
\end{eqnarray*}
\end{proof}

The tile complexity of our construction is the sum of the tile complexities
of all of the gadgets that assemble all extraction regions.

We will first analyze the tile complexity of the extract bit gadgets. Recall that these gadgets convert a $k$-bit binary string, encoded as a strength-$1$ glue, into a one-bit-per-bump representation. The first bit-extraction gadget accepts a $k$-bit binary string, converts the most significant bit of the block into the appropriate bump and then outputs a $(k-1)$-bit binary string. The latter is the input for the second bit-extraction gadget. This process is iterated $k$ times (once for each bit). For a given $n$ and our choice of $k$ (as described above), the number of distinct extract bit gadgets needed in our construction can be computed as:
\begin{eqnarray*}
2\left(\left|\{0,1\}^0\right| + \left|\{0,1\}^1 \right| + \cdots + \left|\{0,1\}^{k-1}\right|\right) & = & 2\left(1+2+\cdots+2^{k-1}\right) \\
							& = & 2\left(2^k - 1\right) < 2\cdot 2^k < 4n/\log n = O\left(n/\log n\right)
\end{eqnarray*}
Since each bit-extraction gadget is comprised of $O(1)$ unique tile types, the total tile complexity for the extraction gadgets is $O(n / \log n)$.

It is easy to see that all other gadgets in our construction can be
implemented using $O(m)$ unique tile types (see Section~\ref{app:optimal_encoding} for details). Thus, by
Lemma~\ref{lem:m_technical_lemma}, the tile complexity of our
construction is $O(n/\log n)$, which is optimal for all
algorithmically random values of $n$.

\section{Optimal self-assembly of squares at temperature 1 in 3D}
\label{sec:construction}
In this section, we describe how to use our 3D temperature 1 optimal encoding construction to prove the following theorem.

\begin{theorem}
\label{thm:main-theorem}
$K^{1}_{3DSA}(N) = O\left(\frac{\log N}{\log \log N}\right)$.
\end{theorem}

\begin{proof}
Our proof is constructive. Figure~\ref{fig:square_construction}a
shows how we build an $N \times N$ square using two counters C1 and C2
and two filler regions F1 and F2. Counter C1 is a zig-zag counter
whose construction is described in the appendix (see
Section~\ref{app:counter}).  Counter C2 is identical to C1 after a
90-degree clockwise rotation. Each counter is seeded with a value
produced by an optimal encoding region (OER for short). The full
construction for F1 is depicted in
Figure~\ref{fig:square_filler_construction}. F2 is a smaller,
mirror-image of F1 with minor modifications to properly connect all of
the pieces of the square. Both F1 and F2 are essentially squares,
except for two hooks needed to stop the horizontal and vertical
growths of each filler region, namely, one eight-tile hook encroaching
on and another one-tile hook protruding from each filler region (see
Figure~\ref{fig:square_filler_construction}). These hooks require
simple modifications of the OER regions (see
Figure~\ref{fig:square_construction}b) that are all located in
the hard-coded (i.e., first and last) block extracting regions of OER1
and OER2. Note that F1 is also missing a two-tile wide rectangle
region on its left that is used up by the vertical connector that
initiates the assembly of OER2 immediately after the assembly of C1
terminates. Figure~\ref{fig:square_construction}c shows the assembly
sequence for the whole square, while
Figure~\ref{fig:square_construction}b zooms in on the region of the
square where OER1, F1, OER2 and F2 all interact.

\begin{figure}[htp]
\begin{minipage}{0.4\textwidth}
\centering
 \includegraphics[width=0.7\textwidth]{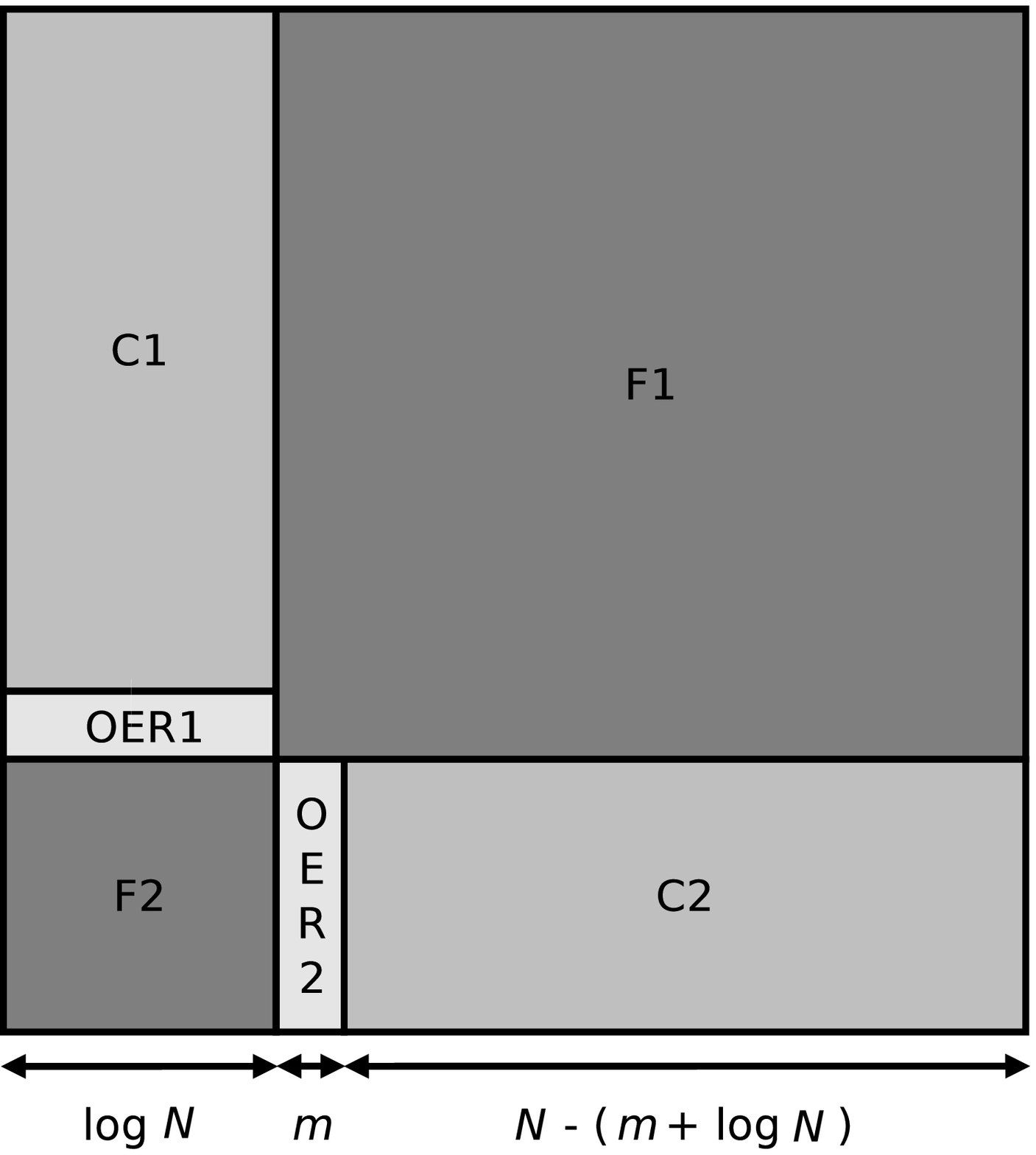}

(a) Overall square construction.
\bigskip

 \includegraphics[width=0.7\textwidth]{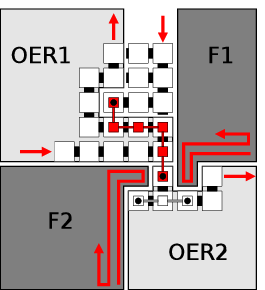}

(b) Detail of the region of the square where OER1, F1, OER2 and F2 meet.
\end{minipage}
\begin{minipage}{0.6\textwidth}
{\centering
 \includegraphics[width=\textwidth]{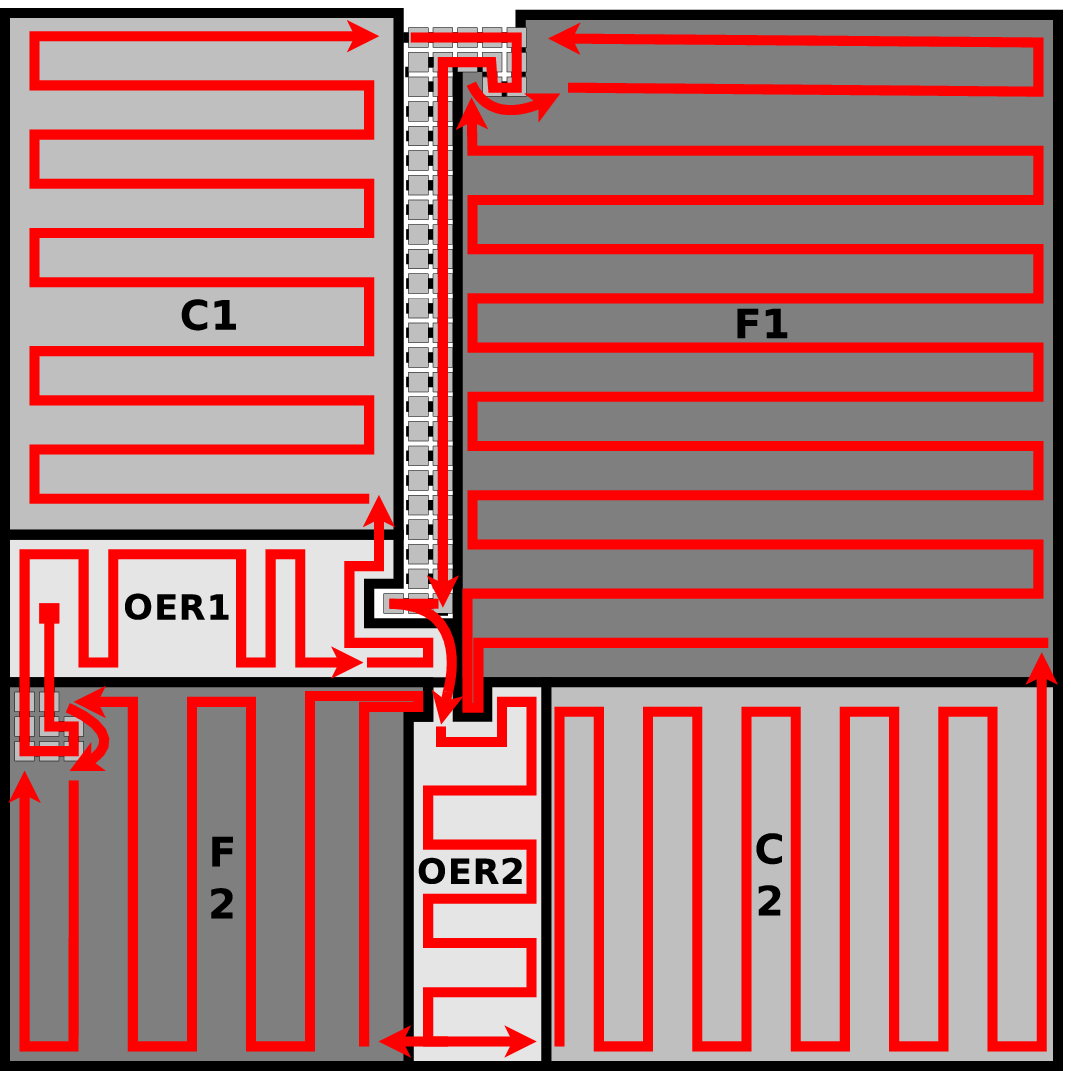}
}

(c) The assembly sequence for the whole square is shown with red
 arrows starting from the seed tile (the red square located in
 OER1). The central region in this sub-figure is shown in more detail in
 Figure~\ref{fig:square_construction}b to the left.
\end{minipage}

\caption{Construction of an $N \times N$ square, where $m$ is
  $O\left(\frac{\log N}{\log \log N}\right)$.  The counters C1 and C2 (in medium
  gray) are identical up to rotation. So are their seed rows, each of
  which is the output of an optimal extraction region (OER1 and OER2,
  respectively, in light gray). F1 and F2 (in dark gray) are filler regions.}
\label{fig:square_construction}
\vspace{0pt}
\end{figure}

\begin{figure}[htp]
\begin{minipage}{\textwidth}
\begin{minipage}{0.4\textwidth}
\centering
 \includegraphics[width=\textwidth]{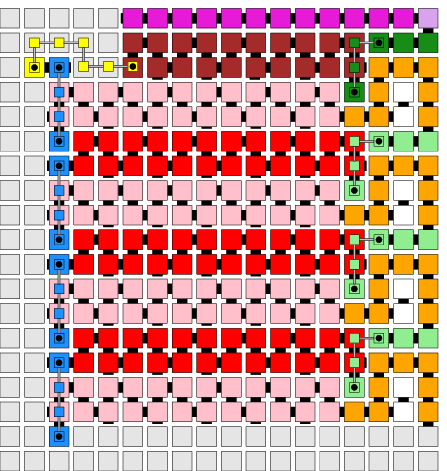}
\end{minipage}\hfill
\begin{minipage}{0.58\textwidth}
 The gray tiles in this figure
do not belong to F1. They are all added to the $N\times N$ square assembly before F1
starts assembling and they determine the height and width of F1, both
of which are adjustable in the following way:\vspace*{-2mm}
\begin{itemize}
\item The height of F1 is always a multiple of four (i.e., the total height
  of each pink plus red gadget), plus the number of purple rows at the
  top, which can be hard-coded to any value in $\{1,2,3,4\}$, together
  with a corresponding increase in the height of the top-left
 (gray)   hook.\vspace*{-2mm}
\item The width of F1 is always a multiple of two (i.e., the width of
  each pink gadget) plus the width of each orange gadget, which is
  either three (by deleting the column occupied by the white tiles) or
  four (as shown).
\end{itemize}\vspace*{-2mm}
Therefore, this construction gives us two knobs, namely the number of
purple rows and the width of the orange gadget, to assemble filler
regions of any height and width, respectively.
\end{minipage}
\end{minipage}
\caption{Detailed construction for the F1 filler region in
  Figure~\ref{fig:square_construction}}\vspace*{-2mm}

\label{fig:square_filler_construction}
\vspace{0pt}
\end{figure}

First, we compute the tile complexity of our construction as the sum
of the tile complexities of all of the components that make up the
$N\times N$ square. Let $n=\lceil \log N \rceil$. If $k$ denotes the
smallest integer satisfying $2^k \geq n/\log n$ and $m$ is defined as
$\lceil n/k\rceil$, then the tile complexity of each OER is $O(n/\log
n)$, as proved in Section~\ref{sec:OER_tile_complexity}. Furthermore,
the tile complexity of each binary counter is $O(1)$ (see
Section~\ref{app:counter}). Finally, the tile complexity of each
filler region is $O(1)$, since each colored gadget in
Figure~\ref{fig:square_filler_construction} has tile complexity
$O(1)$. Therefore, the tile complexity of our square construction is
dominated by that of the OERs and is therefore $O(\frac{n}{\log n})=O\left(\frac{\log
  N}{\log \log N}\right)$.

Second, we need to prove that our tile system is directed and does
produce an $N\times N$ square. The assembly sequence depicted in
Figure~\ref{fig:square_construction}c demonstrates that our tile
system uniquely produces a square. To make sure that this square has
width $N$, we need to pick the initial value $i$ of the counters and
adjust the size of the filler regions as follows. The width of OER1,
C1 and F2 in our construction, and thus also the height of OER2, C2
and F2 is $6n+4$. The height of OER1, and thus also the width of OER2,
is $2m+7$ (see, for example, Figure~\ref{fig:twelve}). Therefore, the
height of C1 (and thus also the width of C2) must be equal to $N -
(2m+7 + 6n+4)= N-2\left\lceil \frac{\lceil \log
  N\rceil}{k}\right\rceil - 6\lceil \log N\rceil - 11$. Let us denote
this value by $h(N)$. Our construction in Section~\ref{app:counter}
gives us two knobs to control the height of any $n$-bit counter: the
initial value $i$ of the counter and the number $r$ of rooftop rows,
where $r \in\{1,2,3,4\}$. Since each value from $i$ to the final value
of the counter $2^n-1$ (inclusive) takes up four rows of tiles, we
must have $\lfloor\frac{h(N)}{4}\rfloor=2^n-i$ and $r=1+h(N) \bmod
4$. Therefore, for both C1 and C2, the initial value of the counter is
$2^n-\lfloor\frac{h(N)}{4}\rfloor$. Finally, the
correct height and width of F1 are obtained by setting the two knobs
described in Figure~\ref{fig:square_filler_construction} to
$1+(2m+7+h(N)-1)\bmod 4$ and $4-(2m+7+h(N)-2)\bmod 2$,
respectively. Similarly, the correct height and width of F2 are
obtained by setting the second knob to $4-(6n+4)\bmod 2$ and the first
knob to $1+(6n+4-1)\bmod 4$.
\end{proof}

\section{Conclusion}
In this paper, we developed a 3D temperature 1 optimal encoding construction, based on the 2D temperature 2 optimal encoding construction of Soloveichik and Winfree \cite{SolWin07}. We then used our construction to answer an open question of Cook, Fu and Schweller \cite{CooFuSch11}, namely, we proved that $K^1_{3DSA}(N) = O\left(\frac{\log N}{\log \log N}\right)$, which is the optimal tile complexity for all algorithmically random values of $N$.

We propose a future research direction as follows. Consider a generalization of the aTAM, called the \emph{two-handed}~\cite{Versus} (a.k.a., hierarchical \cite{CheDot12}, q-tile, multiple tile \cite{AGKS05g}, polyomino \cite{Luhrs08}) abstract Tile Assembly Model (2HAM). A central feature of the 2HAM is that, unlike the aTAM, it allows two ``supertile'' assemblies, each consisting of one or more tiles, to bind. In the 2HAM, an assembly of tiles is producible if it is either a single tile, or if it results from the stable combination of two other producible assemblies. Now define the quantity $K^\tau_{3D2SA}(N)$, as the minimum number of distinct 3D tile types required to uniquely produce it in the 2HAM at temperature $\tau$. An interesting problem would be to determine if $K^1_{3D2SA}(N) = O\left(\frac{\log N}{\log \log N}\right)$.

\newpage

\bibliographystyle{amsplain}
\bibliography{tam}

\newpage

\section{Appendix}
In this section, we provide details on two of our constructions,
namely the optimal encoding and the zig-zag counter, that did not fit
in the main body of the paper.

Throughout this section, we assume that the glues on each tile are implicitly defined to ensure deterministic assembly.

\subsection{Optimal encoding construction}\label{app:optimal_encoding}
The assembly of an extraction region is initiated by the
``block-number'' gadget (see Figure~\ref{fig:fig_14}b). The
block-number gadget is initiated in one of two ways. On the one hand,
the final tile placed in the hard-coded path of tiles that assembles
the first extraction region may initiate the assembly of the
block-number gadget. On the other hand, the block-number gadget may be
initiated by the final tile placed by the final gadget to assemble in
the previous extraction region, known as the ``floor'' gadget, which
we describe later (see Figure~\ref{fig:fig_28}a). The block-number
gadget uses a zig-zag assembly pattern, i.e., it grows in an
alternating left-to-right and right-to-left pattern. This zig-zag
growth pattern is five tiles wide because of the spacing requirements
of our construction. Each zig-zag represents an attempt to place an
``extraction tile'' for a $k$-bit block at a special location, which
is denoted by the dotted tiles above the red tile in the block-number
gadget in Figure~\ref{fig:fig_14}b. When an extraction tile is placed,
the corresponding bits for that specific $k$-bit string are extracted
by a subsequent series of gadgets. The extraction tile for each
$k$-bit block is obstructed (i.e., geometrically hindered from being
placed) by a previous portion of the construction, except for the
extraction tile for the current $k$-bit block (highlighted in red in
Figure~\ref{fig:fig_14}b). Finally, the block-number gadget is
prevented from growing down any further by the tiles placed below it.
These tiles either are part of the perimeter of the first (hard-coded)
extraction region or they were placed by a gadget called the
``repeating hook gadget'' (described below) during the assembly of the
previous extraction region. The placement of these blocking tiles
ensures that only the correct $k$-bit block is extracted within a
given extraction region. Note that we use the same block-number gadget
to initiate the extraction of each of the $m$ distinct $k$-bit blocks,
whence its tile complexity is $O(m)$.

\begin{figure}[htp]%
\doublefigure{65mm}{7mm}{}{The shaded gadget in the extraction region represents the block-number gadget, pictured on the right in Figure~\ref{fig:fig_14}b.}{}{}{The block-number gadget.}{0.55}
    \caption{\label{fig:fig_14} Overview of the block-number gadget and its location in the construction. }%
\end{figure}

The zig-zag pattern of assembly described above for the block-number gadget is a common property of many of the gadgets in our constructions. The basic idea is that, as the zig-zag path assembles, it will first ``zig'' in one direction, for a small constant number of tiles (depending on the spacing requirements of the gadget) and then it will ``zag'' back in the opposite direction. The final tile in the zag direction tries to grow the zag portion of the path by one more tile (in the zag direction), and also initiates the next zig-zag iteration. In all but one case, the extra zag tile is prevented from being placed (i.e., a tile from a previous portion of the construction is already placed at this location) and the zig-zag pattern simply continues. However, in one case, the zig-zag pattern is blocked and there is an empty location at which the extra zag tile is placed. This non-cooperative assembly algorithm allows various gadgets in our construction to ``know'' when they have reached a certain ``stopping point''. Moreover, most of the time, the gadgets that implement this zig-zag assembly algorithm can be implemented in $O(1)$ tile complexity. Therefore, throughout the following discussion, unless noted otherwise, all gadgets are implemented using $O(1)$ unique tile types.

\begin{figure}[htp]%
\doublefigure{47mm}{7mm}{}{The shaded gadget in the extraction region represents the block-hook gadget, pictured on the right in Figure~\ref{fig:fig_15}b.}{}{}{The block-hook gadget.}{0.9}
 \caption{\label{fig:fig_15} Overview of the block-hook gadget and its location in the overall construction. }%
\end{figure}

\begin{figure}[htp]%
\doublefigure{56mm}{7mm}{}{The shaded gadget in the  extraction region represents the up-extraction gadget, pictured on the right in Figure~\ref{fig:fig_16}b.}{}{}{The up-extraction gadget.}{0.25}
    \caption{\label{fig:fig_16} Overview of the up-extraction gadget and its location in the overall construction. }%
\end{figure}

The block-number gadget places the correct extraction tile based on
the block number position, which is encoded in the geometry of a
previous portion of the construction. The placement of an extraction
tile either initiates the assembly of the ``block-hook'' gadget (see
Figure~\ref{fig:fig_15}b) or initiates a different gadget that
assembles the perimeter for the final extraction region via a
hard-coded path of tiles. First, the block-hook gadget assembles an
L-shape path above a portion of the previous extraction region, in the
$z=1$ plane, and then builds, in the $z=0$ plane, a geometric pattern
of tiles in the shape of a one-tile-wide hook that will later stop the
downward growth of the ``hook-seeking gadget'' (see
Figure~\ref{fig:fig_21}b). Then the block-hook gadget initiates the
assembly of the ``up-extraction'' gadget (see
Figure~\ref{fig:fig_16}b). Finally, the block-hook gadget is
responsible for determining the starting point (in the north-south
direction) for a subsequent gadget called the ``hook-initiating''
gadget (see Figure~\ref{fig:fig_22}b). Note that, as the construction
continues, from $k$-bit block to $k$-bit block (except for the first
and last blocks, which are hard-coded), the hook of tiles initially
assembled by the block-hook gadget is translated up by two tiles (one
translation per extraction region). The location of this hook
essentially represents which $k$-bit block to extract next.

As mentioned above, the final tile of the block-hook gadget initiates
the up-extraction gadget. The
up-extraction gadget assembles upward, in a zig-zag pattern (as
described above), parallel to the block-number gadget. The top-left
tile of each zig-zag pattern tries to grow left but is blocked by a
portion of the block-number gadget (this is the zag portion of the
path). When the up-extraction gadget grows to the row immediately
above the top row of the block-number gadget, the red tile shown in
Figure~\ref{fig:fig_16}b is placed. This tile initiates the
``extraction-jump'' gadget (see
Figure~\ref{fig:fig_17}b). Subsequently, the upward growth of the
up-extraction gadget is blocked by tiles from the previous extraction
region. The extraction-jump gadget grows a path of tiles in the $z=1$
plane, above the up-extraction gadget and to the right of a portion of
the previous extraction region. In our construction, we have $O(m)$
distinct block-hook gadgets, up-extraction gadgets and extraction-jump
gadgets (one for each $k$-bit block). The final tile of the
extraction-jump gadget initiates the ``extraction'' gadget (see
Figure~\ref{fig:fig_18}a).

\begin{figure}[htp]%
\doublefigure{47mm}{10mm}{}{The shaded gadget in the  extraction region represents the extraction-jump extraction gadget, pictured on the right in Figure~\ref{fig:fig_17}b.}{}{}{The extraction-jump gadget.}{0.28}
    \caption{\label{fig:fig_17} Overview of the extraction-jump gadget and its location in the overall construction. }%
\end{figure}

\begin{figure}[htp]%
\doublefigurestacked{23mm}{7mm}{}{The shaded gadget in the extraction region represents the extraction gadget, pictured above in Figure~\ref{fig:fig_18}a.}{}{}{The extraction gadget (see Figure~\ref{fig:zoom_extraction_gadget} for a detailed view of the large squares in this figure).}{0.8}{3mm}
    \caption{\label{fig:fig_18} Overview of the extraction gadget and its location in the overall construction. }%
\end{figure}

The extraction gadget extracts the current $k$-bit block into a
one-bit-per-bump representation along the top of the current
extraction region. A bump in the $z=0$ plane is representative of a 0
and a bump in the $z=1$ plane is representative of a 1. These bumps
can be seen clearly in Figure~\ref{fig:zoom_extraction_gadget}. Note
that the extraction gadget is the result of concatenating $k$
``bit-extraction'' gadgets together (see
Figure~\ref{fig:zoom_extraction_gadget}). The bit-extraction gadgets
are ``temperature 1'' versions of the ``extract bit'' (temperature 2)
tile types from Figure 5.7a of \cite{SolWin07}. The tile complexity of
the extraction gadget is $O(m)$ (see our discussion at the end of
Section~\ref{sec:OER_tile_complexity}) and we use the same extraction
gadget to extract all $k$-bit blocks.

After the extraction gadget finishes extracting the $k$ bits of the
current block, it initiates the assembly of the ``ceiling'' gadget
(see Figure~\ref{fig:fig_20}a). The ceiling gadget assembles a path of
tiles from right to left, placing its final tile under the starting
point of the extraction gadget (see the red tile in
Figure~\ref{fig:fig_20}a). Note that, as it assembles toward its
ending point, the ceiling gadget places a tile at a carefully chosen
location in the $z=1$ plane, the purpose of which is to block a
portion of a subsequent gadget, known as the ``repeating-up'' gadget
(see Figure~\ref{fig:fig_23}b), which we describe later. Due to
spacing constraints, this special $z=1$ tile is placed in the column
of tiles that is one tile to the left of the penultimate bit-bump of
the current extraction region. The placement of this special $z=1$
tile signals a subsequent gadget to assemble the remaining perimeter
of the current extraction region and then initiate the assembly of the
next extraction region (the gadgets that carry out these tasks will be
discussed below). We use a single ceiling gadget in our construction
(i.e., the same one is used in all of the $m - 2$
generally-constructed extraction regions) and its tile complexity is
proportional to the width of an extraction region, which is $O(k)$ and
thus $O(m)$. The final tile in the ceiling gadget initiates the
assembly of the ``hook-seeking'' gadget (see
Figure~\ref{fig:fig_21}b).

\begin{figure}[htp]%

\begin{minipage}[t]{\textwidth}
\fbox{
\begin{minipage}[t][15mm][b]{0.45\textwidth}
\begin{center}
    \subfloat[][Bump representing a 0 in the extraction gadget of Figure~\ref{fig:fig_18}a.]{%
        \label{fig:zoom_zero_bump}%
        \includegraphics[width=0.8\textwidth]{}}%
\end{center}
\end{minipage}
}\hfill
\fbox{
\begin{minipage}[t][15mm][b]{0.45\textwidth}
\begin{center}
    \subfloat[][Bump representing a 1 in the extraction gadget of Figure~\ref{fig:fig_18}a.]{%
        \label{fig:zoom_one_bump}%
        \includegraphics[width=0.8\textwidth]{}}
\end{center}
\end{minipage}
}
\end{minipage}
    \caption{\label{fig:zoom_extraction_gadget} Magnified view of the bit-extraction gadget to clarify the differences between bumps representing 0s and 1s. }%
\end{figure}

\begin{figure}[htp]%
\doublefigurestacked{22mm}{3mm}{}{The shaded gadget in the extraction region represents the ceiling gadget, pictured above in Figure~\ref{fig:fig_20}a.}{}{}{The ceiling gadget.}{0.8}{3mm}
    \caption{\label{fig:fig_20} Overview of the ceiling gadget and its location in the overall construction. }%
\end{figure}

\begin{figure}[htp]%
\doublefigure{67mm}{7mm}{}{The shaded gadget in the  extraction region represents the hook-seeking gadget, pictured on the right in Figure~\ref{fig:fig_21}b.}{}{}{The hook-seeking gadget.}{0.25}
    \caption{\label{fig:fig_21} Overview of the hook-seeking gadget and its location in the construction. }%
\end{figure}

The hook-seeking gadget starts growing from the final tile that was placed by the ceiling gadget and grows down in a zig-zag pattern, parallel to the previous up-extraction gadget. The hook-seeking gadget grows down until it finds the location of the block hook (see Figure~\ref{fig:fig_15}b). The assembly of the hook-seeking gadget is very similar to the block-number gadget, i.e., the placement of certain tiles is blocked until the block hook is reached, at which point the downward growth of the hook-seeking gadget is blocked and a special tile is allowed to be placed to the left of the hook-seeking gadget (in the space directly above the block-hook gadget). Once placed, this special tile initiates the assembly of the ``hook-initiating'' gadget (see Figure~\ref{fig:fig_22}b). The same hook-seeking gadget is used in all generally-constructed extraction regions.

\begin{figure}[htp]%
\doublefigure{48mm}{7mm}{}{The shaded gadget in the  extraction region represents the hook-initiating gadget, pictured on the right in Figure~\ref{fig:fig_22}b.}{}{}{The hook-initiating gadget.}{0.8}
    \caption{\label{fig:fig_22} Overview of the hook-initiating gadget and its location in the construction. }%
\end{figure}

First, the hook-initiating gadget assembles a path of tiles in the $z=1$ plane directly above a portion of the previous hook-seeking gadget. Then, it assembles a group of tiles in the shape of a two-tile-wide hook in the $z=0$ plane. This hook of tiles will block the downward assembly of the subsequent ``repeating-down'' gadget (see Figure~\ref{fig:fig_25}b). When the first repeating-down gadget in the current extraction region is blocked by the hook-initiating gadget, the former will ``know'' the location of the latter and thus will initiate the assembly of a translated version of the hook (translated up by two tiles). The same hook-initiating gadget is used in all generally-constructed extraction regions. The final tile of the hook-initiating gadget initiates the assembly of the ``repeating-up'' gadget (Figure~\ref{fig:fig_23}b).

\begin{figure}[htp]%
\doublefigure{53mm}{7mm}{}{The shaded gadget in the  extraction region represents the repeating-up gadget, pictured on the right in Figure~\ref{fig:fig_23}b.}{}{}{The repeating-up gadget.}{0.35}
    \caption{\label{fig:fig_23} Overview of the repeating-up gadget and its location in the construction. }%
\end{figure}

The repeating-up gadget is three tiles wide and uses a zig-zag pattern
of assembly to search for the top of the previous repeating-down
gadget (or the top of the hook-seeking gadget in the case of the first
occurrence of the repeating-up gadget). This top is found when the
top-left tile of the zig-zag pattern can place a special tile in one
of the locations denoted with a dotted outline (under the red tile) in
Figure~\ref{fig:fig_23}b. Further upward zig-zag growth of the
repeating-up gadget is blocked by the ceiling gadget. The tile
two tiles south of the big red tile in Figure~\ref{fig:fig_23}b grows north
one location in the $z=0$ plane and is forced to make a decision:
either (1) assemble into the $z=1$ plane and initiate the
``initiate-repeating-down'' gadget (see Figure~\ref{fig:fig_24}b), or
(2) if such growth is blocked in the $z=1$ plane, grow north one more
location and initiate the assembly of the
``initiate-next-extraction-region'' gadget (see
Figure~\ref{fig:fig_26}b). The same repeating-up gadget is used in
all generally-constructed extraction regions.

\begin{figure}[htp]%
\doublefigure{48mm}{12mm}{}{The shaded gadget in the  extraction region represents the initiate-repeating-down gadget, pictured on the right in Figure~\ref{fig:fig_24}b.}{}{}{The initiate-repeating-down gadget.}{0.45}
    \caption{\label{fig:fig_24} Overview of the initiate-repeating-down gadget and its location in the construction. }%
\end{figure}

The assembly of the initiate-repeating-down gadget is initiated by the repeating-up gadget. It is basically a line of tiles that assembles in the $z=1$ plane and ``jumps'' over the top row of the previous repeating-up gadget. An important property of the initiate-repeating-down gadget is that it can only form when a certain tile in the previous repeating-up gadget is not blocked (by a tile in the ceiling gadget) in the $z=1$ plane. The final tile of the initiate-repeating-down gadget initiates the assembly of another repeating-down gadget. The same initiate-repeating-down gadget is used in all generally-constructed extraction regions.

\begin{figure}[htp]%
\doublefigure{53mm}{7mm}{}{The shaded gadget in the extraction region represents the repeating-down gadget, pictured on the right in Figure~\ref{fig:fig_25}b.}{}{}{The repeating-down gadget.}{0.35}
    \caption{\label{fig:fig_25} Overview of the repeating-down gadget and its location in the construction. }%
\end{figure}

The assembly of a repeating-down gadget is initiated by the final tile
of the initiate-repeating-down gadget. The purpose of the
repeating-down gadget is to find either the repeating-hook gadget or
the hook-initiating gadget (note that the latter scenario only occurs
with the first repeating-down gadget within each extraction
region). When the hook is found, the repeating-down gadget places a
tile at a special location, namely one of the locations denoted by a
dotted outline above the red tile in Figure~\ref{fig:fig_25}b. The red
tile placed at this special location initiates the assembly of the
``repeating-hook'' gadget (see Figure~\ref{fig:fig_27}b). The same
repeating-down gadget is used in all generally-constructed
extraction regions.

\begin{figure}[htp]%
\doublefigure{48mm}{7mm}{}{The shaded gadget in the extraction region represents the repeating-hook gadget, pictured on the right in Figure~\ref{fig:fig_27}b.}{}{}{The repeating-hook gadget.}{0.85}
    \caption{\label{fig:fig_27} Overview of the repeating-hook gadget and its location in the construction. }%
\end{figure}

The repeating-hook gadget assembles a path of tiles in the $z=1$ plane in order to avoid a portion of the previous repeating-down gadget. This $z=1$ path assembles right and then down and ultimately assembles a group of tiles in the shape of a hook in the $z=0$ plane, similar to the shape of the hook-initiating gadget. The repeating-hook gadget is initiated by each repeating-down gadget. The main purpose of the repeating-hook gadget is to geometrically propagate the block number position, via the hook of tiles, through the current extraction region. The hook shape of the repeating-hook gadget will also serve to block the downward assembly of the next repeating-down gadget. Note that the final, rightmost hook within an extraction region will serve to block the downward assembly of the block-number gadget of the next extraction region. The same repeating-hook gadget is used in all generally-constructed extraction regions.

\begin{figure}[htp]%
\doublefigure{48mm}{12mm}{}{The shaded gadget in the extraction region represents the initiate-next-extraction-region gadget, pictured on the right in Figure~\ref{fig:fig_26}b.}{}{}{The initiate-next-extraction-region gadget.}{0.7}
    \caption{\label{fig:fig_26} Overview of the initiate-next-extraction-region gadget and its location in the construction. }%
\end{figure}

In the case where the repeating-up gadget is blocked in the $z=1$ plane (by a particular tile in the ceiling gadget, as described above), it cannot initiate the assembly of another repeating-down gadget. However, in this case, because of the geometry of the ceiling gadget, the repeating-up gadget may grow a path of tiles in the $z=0$ plane up and underneath part of the ceiling gadget, much like how a highway runs directly underneath an overpass. This is essentially the ``signal'' from the ceiling gadget to the repeating-up gadget that the current extraction region is almost completed. Note that this signal is hard-coded into the geometry of the ceiling gadget for every extraction region (an obvious consequence of the fact that we use a single ceiling gadget in all generally-constructed extraction regions). The red tile in Figure~\ref{fig:fig_23}b, in this case, is blocked from growing into the $z=1$ plane, but is unblocked on its north side in the $z=0$ plane and therefore initiates the assembly of the ``initiate-next-extraction-region'' gadget (see Figure~\ref{fig:fig_26}b).

\begin{figure}[htp]%
\doublefigurestacked{82mm}{7mm}{}{The shaded gadget in the extraction region represents the floor gadget, pictured above in Figure~\ref{fig:fig_28}a.}{}{}{The floor gadget.}{0.9}{3mm}
    \caption{\label{fig:fig_28} Overview of the floor gadget and its location in the overall construction. }%
\end{figure}

The initiate-next-extraction-region gadget is a short horizontal path of tiles in the $z=0$ plane, the last of which initiates the assembly of the ``floor'' gadget (see Figure~\ref{fig:fig_28}a). The same initiate-next-extraction-region gadget is used in all generally-constructed extraction regions.

The assembly of the floor gadget is initiated by the last tile placed by the initiate-next-extraction-region gadget. The floor gadget serves two purposes: it first places tiles along the bottom row of the current extraction region and then it initiates the assembly of the next extraction region by initiating the block-number gadget for the next extraction region. See Figure~\ref{fig:fig_28}a for an example of how the floor gadget assembles. Since this gadget must assemble a path of tiles of length $O(\textmd{perimeter of an extraction region})$, its tile complexity is $O(m)$. We use the same floor gadget in all generally-constructed extraction regions.

\subsection{Construction of a zig-zag counter}\label{app:counter}

In this section, we describe the construction for the binary, $n$-bit,
zig-zag counter that we use in our square construction.\footnote{Our
  construction is different from the one described
  in~\cite{CooFuSch11}. That paper describes a general procedure for
  converting any 2D temperature 2 zig-zag tile system into a 3D
  temperature 1 tile system. For example, one difference is in the
  scaling factor in the vertical dimension, that is, how many rows of
  tiles are needed in the temperature 1 tile system to represent a
  single increment row or copy row in the temperature 2 tile
  system. In our construction, this scaling factor is equal to 2,
  while it is equal to 4 in the conversion procedure described
  in~\cite{CooFuSch11}. Of course, our construction only produces
  binary counters and does not apply to any other zig-zag tile
  system.} An example assembly with $n=3$ is depicted in
Figure~\ref{fig:counter}. The initial value of the counter is encoded
as a geometric pattern of bit-bumps. This seed row, which is part of
another construction (in our case, an optimal encoding region or OER),
appears as bit-bumps sticking out on the north side of the row of gray
tiles at the bottom of Figure~\ref{fig:counter}. In our example, the
initial value of the counter is 000. The assembly of the counter
starts at the single north glue drawn in orange and sticking out of the
OER in the bottom-right corner of the figure. The assembly proceeds by
alternating increment rows, assembling from right to left (in blue in
the figure) and copy rows, assembling from left to right (in green in
the figure). The counter stops when the maximum $n$-bit value is
reached, at which point it assembles one additional increment row (in
blue) and one flat roof (i.e., with no bumps on the north side), shown
as white tiles in the figure. The tile complexity of this
construction, which is described in detail in the rest of this
section, is $O(1)$.

\begin{figure}[ht]
\centering
 \includegraphics[width=0.8\textwidth]{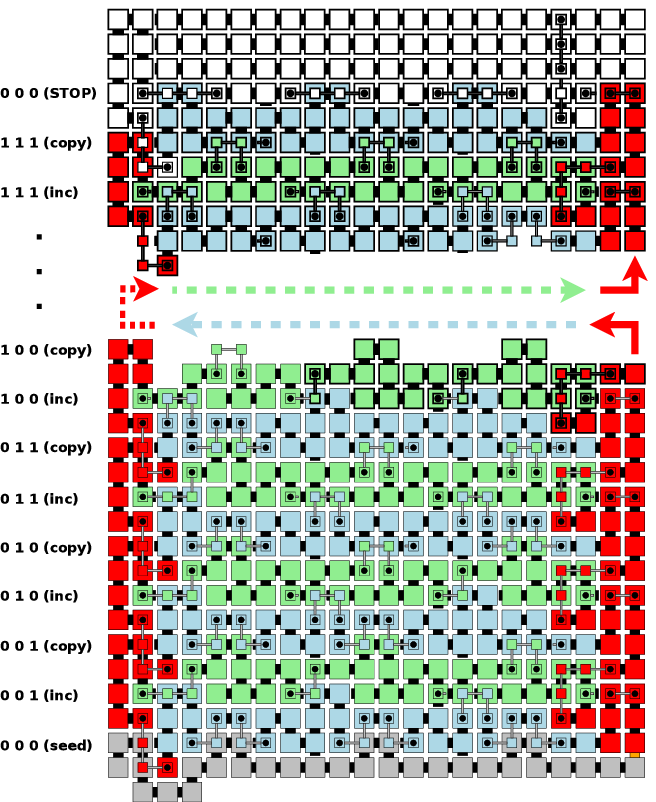}
\caption{Counter construction. The gray tiles at the bottom of the
  counter are part of the optimal extraction region that produces the
  seed value. All other tiles are part of the gadgets shown in
  Figures~\ref{fig:counter_gadgets1}
  and~\ref{fig:counter_gadgets2}. The assembly of the counter starts
  at the orange glue in the bottom-right corner of the figure.
}
\label{fig:counter}
\vspace{0pt}
\end{figure}

The counter construction begins with a right wall, that is, the gadget
depicted in Figure~\ref{fig:counter_gadgets1}b, that will serve to
block the growth of the next copy row. But first, the right-wall
gadget initiates an increment row.

The three gadgets needed for the increment rows are shown at the
bottom of Figure~\ref{fig:counter_gadgets1}. The main gadget in this
group, depicted in Figure~\ref{fig:counter_gadgets1}(e), increments
each bit (from 0 to 1 or from 1 to 0). Note that the bit advertised by
the previous row is not only incremented but also shifted by two tiles
to the left.  The second gadget in this group is the copy gadget
depicted in Figure~\ref{fig:counter_gadgets1}(d). This gadget is used
to leave the bits unchanged in the increment row once the rightmost
0-bit has been incremented and no carry needs to be propagated. Again,
the copied bits are shifted by two tiles to the left. This shift also
happens with the third gadget in this group, depicted in
Figure~\ref{fig:counter_gadgets1}(f), which is specific to the least
significant bit of the counter: the notch (i.e., the two missing tiles
in the top-right corner of the gadget) will serve as the starting
point for a later gadget.

The left-wall gadget, depicted in
Figure~\ref{fig:counter_gadgets1}a, is initiated only when the
bottom-left tile in any one of the increment gadgets is allowed to
grow south. This wall is used to mark the end of the current increment row and
initiate the next copy row.

The two gadgets needed for the copy rows are shown at the top of
Figure~\ref{fig:counter_gadgets2}. Both gadgets in this group copy
each bit (unchanged) and shift the corresponding bit-bump two tiles to
the right to compensate for the leftward shift performed in the
previous increment row. Each copy row starts with the gadget in
Figure~\ref{fig:counter_gadgets2}a that copies the most significant
bit of the counter. The other gadget in this group, depicted in
Figure~\ref{fig:counter_gadgets2}b, copies all other bits. This
gadget also detects the end of the copy row when its bottom-right tile
is allowed to grow south (and is simultaneously blocked in its
rightward movement in the $z=1$ plane by the right wall that was
assembled at the beginning of the previous increment row). At that
point, the right-wall foundation gadget (see
Figure~\ref{fig:counter_gadgets1}(c)) takes over and initiates another
iteration of the increment row/copy row construction by building a
right wall.

The next group of gadgets in our construction are used to detect that
the maximum value has been reached, that is, when all $n$ bits are
equal to 1. These gadgets are modified copies of all of the gadgets
that we have described so far. The only difference between each copy
and the original gadget is that the new gadget remembers that the most
significant bit of the counter has already been incremented to the
value 1. These gadgets are not depicted individually since they are
identical to the red, blue and green gadgets except for, say, a prime
being added to their glue labels. These gadgets, shown with bold
outlines in Figure~\ref{fig:counter}, are used exclusively in the
``top-half'' of the counter construction, or as soon as the ``msb
right copy'' gadget has incremented the most significant bit from 0 to
1 (see the row labelled ``1 0 0 (inc)'' in Figure~\ref{fig:counter}).

To complete the construction of the counter as a perfect rectangle, we
need to build a flat roof on top. This roof construction starts at the
south glue of the bottom-left tile in the ``copied and modified''
(bold) increment gadget. This glue initiates the assembly of the ``msb
eave'' gadget (see Figure~\ref{fig:counter_gadgets2}(c)), which makes
up the topmost left wall and allows the roof to assemble. First, the
``middle bottom roof'' gadget (see
Figure~\ref{fig:counter_gadgets2}(d)) is repeated from left to right
to form a single row of (white) tiles with no bit-bumps on its north
side. Second, the main roof gadget (see
Figure~\ref{fig:counter_gadgets2}(e)) is hard-coded to assemble
between 1 and 4 rows of tiles (again, with a flat top). The height of
this last gadget depends on the target height $h$ of the counter in
the following way: the bottommost row in the main roof gadget rounds
up the total height of the counter to a multiple of 4 (note that the
``middle bottom roof'' row and the bottom row of the main roof gadget
together play the role of the last green, copy row). Then the number
of \emph{additional} rows in the main roof gadget must be equal to $h$
modulo 4. Therefore, the number of full rows of white tiles in the
main roof gadget must be equal to $1 + h \bmod{4}$. Of course, the
height of the ``msb eave'' gadget must also be adjusted to match the
height of the main roof gadget.

In conclusion, this construction uses an $n$-bit counter to build a
rectangle with width $6n+4$ and height $4(2^n-i)+r$, where $i$ is
the initial value of the counter and $r$ is equal to the height of the
counter modulo 4.

\begin{figure}[ht]
\centering


\begin{minipage}[t]{\textwidth}
\centering
\fbox{
\begin{minipage}[t][45mm][b]{0.31\textwidth}
\centering
 \includegraphics[width=0.35\textwidth]{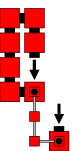}

(a) Left-wall gadget
\end{minipage}}\hfill
\fbox{\begin{minipage}[t][45mm][b]{0.31\textwidth}
\centering
 \includegraphics[width=0.35\textwidth]{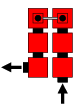}

(b) Right-wall gadget
\end{minipage}}\hfill
\fbox{\begin{minipage}[t][45mm][b]{0.31\textwidth}
\centering
 \includegraphics[width=0.45\textwidth]{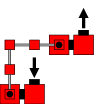}

(c) Right-wall foundation gadget
\end{minipage}}
\end{minipage}

\medskip


\begin{minipage}[t]{\textwidth}
\centering
\fbox{
\begin{minipage}[t][65mm][b]{51.1mm}
\centering
 \includegraphics[width=0.9\textwidth]{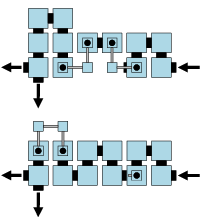}

(d) Left copy gadget\\ (top: configuration for bit 0; \\bottom: configuration for bit 1)
\end{minipage}}\hfill
\fbox{
\begin{minipage}[t][65mm][b]{50mm}
\centering
 \includegraphics[width=0.9\textwidth]{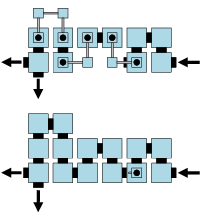}

(e) Increment gadget\\ (top: incrementing the 0 bit;\\ bottom: incrementing the 1 bit)
\end{minipage}}\hfill
\fbox{
\begin{minipage}[t][65mm][b]{50mm}
\centering
 \includegraphics[width=0.9\textwidth]{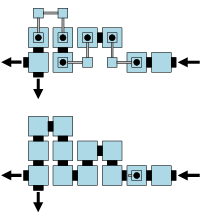}

(f) Lsb increment gadget\\ (top: incrementing the 0 bit;\\ bottom: incrementing the 1 bit)
\end{minipage}}
\end{minipage}

\caption{First set of gadgets used in the construction of the zig-zag counter. In each gadget, the black arrows indicate the entry and exit points of the gadget.}
\label{fig:counter_gadgets1}
\vspace{0pt}
\end{figure}

\begin{figure}[ht]
\centering


\begin{minipage}[t]{\textwidth}
\centering

\fbox{
\begin{minipage}[t][65mm][b]{0.45\textwidth}
\centering
 \includegraphics[width=0.65\textwidth]{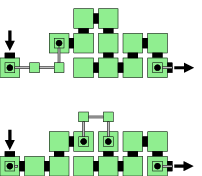}

(a) Msb right copy gadget\\ (top: configuration for bit 0; \\bottom: configuration for bit 1)
\end{minipage}}\hfill
\fbox{
\begin{minipage}[t][65mm][b]{0.45\textwidth}
\centering
 \includegraphics[width=0.7\textwidth]{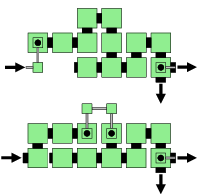}

(b) Right copy gadget\\ (top: configuration for bit 0; \\bottom: configuration for bit 1)
\end{minipage}}
\end{minipage}

\bigskip


\begin{minipage}[t]{\textwidth}
\centering
\fbox{
\begin{minipage}[t][60mm][b]{0.45\textwidth}
\centering
 \includegraphics[width=0.25\textwidth]{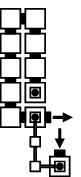}

(c) Msb eave gadget: the two-tile-wide leftmost column can vary in height
 from 2 to 5 tiles (5 tiles are shown above)
\end{minipage}}\hfill
\fbox{
\begin{minipage}[t][60mm][b]{0.45\textwidth}
\centering
 \includegraphics[width=0.7\textwidth]{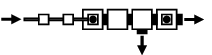}

(d) Middle bottom roof gadget
\end{minipage}}
\end{minipage}

\bigskip

\fbox{
\begin{minipage}[t]{0.978\textwidth}
\centering
 \includegraphics[width=0.7\textwidth]{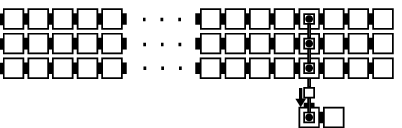}

(e) Main roof gadget: The roof itself can vary in height from 1 to 4 tiles (3 are shown here)
\end{minipage}}

\caption{Second set of gadgets used in the construction of the zig-zag
  counter. In each gadget, the black arrows indicate the entry and
  exit points of the gadget.}
\label{fig:counter_gadgets2}
\vspace{0pt}
\end{figure}

\end{document}